\documentclass[a4paper]{article}

\usepackage{amsmath}
\usepackage{amsthm}
\usepackage{amssymb}

\newtheorem{theorem}{Theorem}
\newtheorem{lemma}{Lemma}
\newtheorem{example}{Example}

\usepackage{authblk}

\title{Approximation ratio of RePair\footnote{The second and third author were supported by the DFG research grant LO 748/10-1.}}
\author[1]{Danny Hucke}
\author[2]{Artur Je\.{z}}
\author[1]{Markus Lohrey}
\affil[1]{University of Siegen, Germany}
\affil[2]{University of Wroclaw, Poland}
\date{}                     %% if you don't need date to appear
\newcommand{\bb}{\mathbb}
\newcommand{\val}{\mathrm{val}}
\newcommand{\mc}{\mathcal}

\begin{document}

\maketitle

\begin{abstract}
In a seminal paper of Charikar et al.~on the smallest grammar problem, the authors derive upper and lower bounds on the approximation
ratios for several grammar-based
compressors.
Here we improve the lower bound for the famous {\sf RePair} algorithm from $\Omega(\sqrt{\log n})$ to $\Omega(\log n/\log\log n)$.
The family of words used in our proof is defined over a binary alphabet, while the lower bound from Charikar et al. needs an alphabet of logarithmic size in the length of the provided words.
\end{abstract}

\section{Introduction}

The idea of grammar-based compression is based on the fact that in many cases a word $w$ can be succinctly
represented by a context-free grammar that produces exactly $w$. Such a grammar is called a {\em straight-line 
program} (SLP) for $w$. In the best case, one gets an SLP of size $O(\log n)$ for a word of length $n$,
where the size of an SLP is the total length of all right-hand sides of the rules of the grammar. A {\em grammar-based
compressor} is an algorithm that produces for a given word $w$ an SLP $\bb A$ for $w$, where, of course, $\bb A$ 
should be smaller than $w$.  Grammar-based compressors can be found at many places in the literature. Probably the best known example 
is the classical {\sf LZ78}-compressor of Lempel and Ziv \cite{ZiLe78}. Indeed, it is straightforward to transform
the {\sf LZ78}-representation of a word $w$ into an SLP for $w$. Other well-known grammar-based compressors
are {\sf BISECTION} \cite{KiefferYNC00}, {\sf SEQUITUR} \cite{Nevill-ManningW97}, and {\sf RePair} \cite{LarssonM99}, just to mention a few. 

One of the first appearances of straight-line programs in the literature are  \cite{BerstelB87,Diw86}, where they are called {\em word chains}
(since they generalize addition chains from numbers to words). In  \cite{BerstelB87}, Berstel and Brlek prove that
the function $g(k,n) = \max \{ g(w) \mid w \in \{1,\ldots,k\}^n \}$, where $g(w)$ is the size of a smallest SLP for the word $w$,
is in $\Theta(n/\log_k n)$. Note that $g(k,n)$ measures the worst case SLP-compression over all words of length $n$ over
a $k$-letter alphabet. The first systematic investigations of grammar-based compressors are \cite{CLLLPPSS05,KiYa00}.
Whereas in \cite{KiYa00}, grammar-based compressors are used for universal lossless compression (in the information-theoretic
sense), Charikar et al.~study in \cite{CLLLPPSS05} the worst case approximation ratio of grammar-based compressors.
For a given grammar-based compressor $\mathcal{C}$ that computes from a given word $w$ an SLP $\mathcal{C}(w)$ for $w$
one defines the approximation ratio of $\mathcal{C}$ on $w$ as the quotient of the size of  $\mathcal{C}(w)$ and 
the size $g(w)$ of a smallest SLP for $w$. The approximation ratio $\alpha_{\mathcal{C}}(n)$ is the maximal approximation
ratio of $\mathcal{C}$ among all words of length $n$ over any alphabet. %\footnote{In \cite{CLLLPPSS05} only
%the approximation ratio $\alpha_{\mathcal{C}}(n,n)$ for an unbounded alphabet is considered.  We believe that 
%it is worthwhile to analyze the dependence of the approximation ratio on the alphabet size.}
In \cite{CLLLPPSS05} the authors compute upper and lower bounds for the approximation ratios of several
grammar-based compressors (among them are the compressors mentioned above). 
The contribution of this paper is the improvement of the lower bound for {\sf RePair} from $\Omega(\sqrt{\log n})$ to $\Omega(\log n/\log\log n)$.
While in \cite{CLLLPPSS05} the lower bound needs an unbounded alphabet (the alphabet grows logarithmically in the length of 
the presented words) our family of words is defined over a binary alphabet.

{\sf RePair} works by repeatedly searching for a digram $d$ (a string of length two) with the maximal number
of non-overlapping occurrences in the current text and replacing all these occurrences by a new nonterminal $A$.
Moreover, the rule $A \to d$ is added to the grammar. {\sf RePair} is one of the so-called global grammar-based compressor from 
\cite{CLLLPPSS05} for which the approximation ratio seems to be very hard to analyze. Charikar et al. prove 
for all global grammar-based compressors an upper bound of $\mathcal{O}\left((n/\log n)^{2/3}\right)$
for the approximation ratio. Note that the gap to our improved lower bound $\Omega(\log n/\log\log n)$
is still large.

\paragraph{Related work.}
The theoretically best known grammar-based compressors with a polynomial (in fact, linear) running time achieve
an approximation ratio of $O(\log n)$ \cite{CLLLPPSS05,Jez15tcs,Jez16,Ryt03}.
In \cite{HuLoRe17}, the precise (up to constant factors) approximation ration for BISECTION (resp., LZ78) was shown to be
$\Theta( (n/\log n)^{1/2})$ (resp., $\Theta( (n/\log n)^{2/3})$).
In \cite{NavarroR08} the authors prove that {\sf RePair} combined with a simple binary encoding of the grammar
compresses every word $w$ over an alphabet of size $\sigma$
to at most $2 H_k(w) + o(|w| \log \sigma)$ bits, for any $k = o(\log_\sigma |w|)$,
where $H_k(w)$ is the $k$-th order entropy of $w$.

There is also a bunch of papers with practical applications for {\sf RePair}: 
web graph compression \cite{ClaudeN10}, bit maps \cite{NavarroPV11}, compressed
suffix trees \cite{Gonzalez07}. Some practical improvements of {\sf RePair} can
be found in \cite{GaJe17}.

\section{Preliminaries}

Let $[1,k] = \{1,\ldots,k\}$. 
Let $w=a_1\cdots a_n$ ($a_1,\dots,a_n\in\Sigma$) be a \emph{word} or \emph{string} over a finite \emph{alphabet} $\Sigma$.
The length $|w|$ of $w$ is $n$ and we denote by $\varepsilon$ the word of length $0$. 
We define $w[i]=a_i$ for $1\le i\le |w|$ and $w[i:j]=a_i\cdots a_j$ for $1\le i\le j\le |w|$.
Let $\Sigma^+ = \Sigma^* \setminus \{\varepsilon\}$
be the set of nonempty words.
For $w \in \Sigma^+$, we call $v\in\Sigma^+$ a \emph{factor} of $w$ if there exist $x,y\in\Sigma^*$ such that $w=xvy$.
If $x=\varepsilon$, then we call $v$ a \emph{prefix} of $w$.
%A factorization of $w$ is a decomposition $w=f_1\cdots f_\ell$ into factors $f_1,\dots, f_\ell$.
For words $w_1,\dots, w_n\in\Sigma^*$, we further denote by $\prod_{i=j}^nw_i$ the word $w_jw_{j+1}\cdots w_n$ if $j\le n$ and $\varepsilon$ otherwise.

A \emph{straight-line program}, briefly SLP, is a context-free grammar that
produces a single word $w\in\Sigma^+$. 
Formally, it is a tuple $\bb A = (N,\Sigma, P, S)$, where $N$ is a 
finite set of nonterminals with $N\cap \Sigma = \emptyset$, 
$S \in  N$ is the start nonterminal, and $P$ is a finite
set of productions (or rules) of the form $A \to w$ for $A \in N$, $w \in (N \cup \Sigma)^+$ such that:
\begin{itemize}
\item For every $A \in N$, there exists exactly one production of the form $A \to w$, and
\item  the binary relation $\{ (A, B) \in N \times N \mid (A \to w) \in P,\;B \text{ occurs in } w \}$ is acyclic.
\end{itemize}
Every nonterminal $A \in N$ produces a unique string $\val_{\bb A}(A) \in \Sigma^+$.
The string defined by $\bb A$ is $\val(\bb A) = \val_{\bb A}(S)$. 
We omit the subscript $\bb A$ when it is clear from the context.
The \emph{size} of the SLP $\bb A$ is 
$|\bb A| = \sum_{(A \to w) \in P} |w|$.
We denote by $g(w)$ the size of a smallest SLP producing the word $w\in\Sigma^+$.
We will use the following lemma:
%An SLP can be transformed
%in linear time into {\em Chomsky normal form}, i.e., for each production $A \to w$, either $w \in \Sigma$
%or $w = BC$ where $B,C \in N$.
%We will use the following lemma which summarizes known results about SLPs.

\begin{lemma}[\mbox{\cite[Lemma~3]{CLLLPPSS05}}] \label{lemma:folklore}
A string $w$ contains at most $g(w) \cdot k$ distinct factors of length $k$.
%If a string $w$ is generated by an SLP of size $m$, then $w$ contains at most $mk$ distinct factors of length $k$.
%
%et $\Sigma$ be a finite alphabet.% of size $k$.
%\begin{enumerate}
%\item\label{nlogn} For every word $w\in\Sigma^+$ of length $n$, there exists an SLP $\bb A$ of size $O(n/\log n)$ such that $\val(\bb A)=w$.
%\item\label{logn}  For an SLP $\bb A$ and a number $n>0$,
%there exists an SLP $\bb B$ of size $|\bb A|+O(\log n)$ such that $\val(\bb B)=\val(\bb A)^n$.
%\item\label{concat} For SLPs $\bb A_1$ and $\bb A_2$ there exists an SLP $\bb B$ of size $|\bb A_1|+|\bb A_2|$ such that $\val(\bb B)=\val(\bb A_1)\val(\bb A_2)$.
%\item\label{substring}  For given words $w_1,\dots,w_n\in\Sigma^*$, $u\in\Sigma^+$ and SLPs
%$\bb A_1, \bb A_2$ with $\val(\bb A_1)=u$ and $\val(\bb A_2)=
%w_1 x w_2 x \cdots w_{n-1} x w_n$ for a symbol $x \not\in \Sigma$,
%there exists an SLP $\bb B$ of size $|\bb A_1| + |\bb A_2|$ such that $\val(\bb B)=w_1 u w_2 u \cdots w_{n-1} u w_n$.
%\end{enumerate}
\end{lemma}
A grammar-based compressor $\mc C$ is an algorithm that computes for a nonempty word $w$ an SLP $\mc C(w)$ such that $\val(\mc C(w))=w$.
%We denote by $\mc C(w)$ the SLP produced by $\mc C$ on input $w$.
The \emph{approximation ratio} $\alpha_{\mc C}(w)$ of $\mc C$ for an input $w$ is defined as $|\mc C(w)|/g(w)$.
The worst-case approximation ratio $\alpha_{\mc C}(k,n)$ of $\mc C$ is the maximal approximation ratio over
all words of length $n$ over an alphabet of size $k$:
\[\alpha_{\mc C}(k,n)=\max \{ \alpha_{\mc C}(w) \mid w \in [1,k]^n \} = \max\{ |\mc C(w)|/g(w) \mid w \in [1,k]^n \} \]
%It is easy to see that $\alpha_{\mc C}(k,n)\ge\alpha_{\mc C}(w)$ for each word $w$ of length $n$ over an alphabet of size $k$.
%We will use this property to establish lower bounds for the worst-case approximation ratio by analysing particular families of words.
If the alphabet size is unbounded, i.e., if we allow alphabets of size $|w|$, then we write $\alpha_{\mc C}(n)$ instead of $\alpha_{\mc C}(n,n)$.

\section{RePair}
For a given SLP $\bb A = (N,\Sigma, P, S)$, a word $\gamma \in (N \cup \Sigma)^+$ is called a \emph{maximal string} of $\bb A$ if
\begin{itemize}
\item $|\gamma|\ge 2$,
\item $\gamma$ appears at least twice without overlap in the right-hand sides of $\bb A$,
\item and no strictly longer word appears at least as many times on the ride-hand sides of $\bb A$ without overlap.
\end{itemize}
A \emph{global grammar-based compressor} starts on input $w$ with the SLP $\bb A=(\{S\},\Sigma, \{S\to w\}, S)$.
In each round, the algorithm selects a maximal string $\gamma$ of $\bb A$
and updates $\bb A$ by replacing a largest set of a pairwise non-overlapping occurrences of $\gamma$ in $\bb A$ by a fresh nonterminal $X$.
Additionally, the algorithm introduces the rule $X\to \gamma$.
The algorithm stops when no maximal string occurs.
The global grammar-based compressor {\sf RePair}~\cite{LarssonM99} selects in each round a most frequent maximal string.
Note that the replacement is not unique, e.g. the word $a^5$ with the maximal string $\gamma=aa$ yields SLPs with rules $S\to XXa, X\to aa$ or $S\to XaX, X\to aa$ or $S\to aXX, X\to aa$.
We assume the first variant in this paper, i.e. maximal strings are replaced from left to right.

The above description of RePair is taken from~\cite{CLLLPPSS05}. In most papers on {\sf RePair} the algorithm works slightly different: It replaces
in each step a digram (a string of length two) with the maximal number of pairwise non-overlapping occurrences in the right-hand sides.  
For example, for the string $w = abcabc$ this produces the SLP $S \to BB$, $B \to Ac$, $A \to ab$, whereas the 
{\sf RePair}-variant from \cite{CLLLPPSS05} produces the smaller SLP $S \to AA$, $A \to abc$.

The following lower and upper bounds on the approximation ratio of {\sf RePair} were shown in~\cite{CLLLPPSS05}:
\begin{itemize}
\item $\alpha_\mathsf{RePair}(n)\in\Omega\left(\sqrt{\log n}\right)$
\item $\alpha_\mathsf{RePair}(2,n)\in \mathcal{O}\left((n/\log n)^{2/3}\right)$
\end{itemize}
The proof of the lower bound in~\cite{CLLLPPSS05} assumes an alphabet of unbounded size.
To be more accurate, the authors construct  for every $k$ a word $w_k$ of length $\Theta(\sqrt{k} 2^k)$ over and alphabet of size 
$\Theta(k)$ such that $g(w) \in O(k)$ and {\sf RePair} produces a grammar of size $\Omega(k^{3/2})$ for $w_k$.               
%over an alphabet of size $\Theta(\log |w_k|)$
%with $\alpha_{\mathsf{RePair}}(w_k)\in\Omega(\sqrt{\log |w_k|})$.
We will improve this lower bound using only a binary alphabet. To do so, we first need to know how {\sf RePair} compresses unary words.
\begin{example}[unary inputs]
\label{unary}
{\sf RePair} produces on input $a^{27}$ the SLP with rules $X_1\to aa$, $X_2\to X_1X_1$, $X_3\to X_2X_2$ and $S\to X_3X_3X_3X_1a$, where $S$ is the start nonterminal. For the input $a^{22}$ only the start rule $S\to X_3X_3X_2X_1$ is different.
\end{example}

In general, {\sf RePair} creates on unary input $a^m$ ($m\ge 4$) the rules $X_1\to aa$, $X_i\to X_{i-1}X_{i-1}$ for $2\le i\le \lfloor \log m\rfloor-1$ and a start rule, which is strongly related to the binary representation of $m$ since each nonterminal $X_i$ produces the word $a^{2^i}$. To be more accurate, let
$b_{\lfloor \log m\rfloor} b_{\lfloor \log m\rfloor-1}\cdots b_1b_0$
be the binary representation of $m$ and define the mappings $f_i$ ($i \geq 0$) by:
\begin{itemize}
\item $f_0:\{0,1\}\to\{a,\varepsilon\}$ with $f_0(1)=a$ and $f_0(0)=\varepsilon$,\label{f0}
\item $f_i:\{0,1\}\to \{X_i,\varepsilon\}$ with $f_i(1)=X_i$ and $f_i(0)=\varepsilon$ for $i\ge 1$. %$1\le i\le \lfloor\log m\rfloor-1$.
\end{itemize}
Then the start rule produced by {\sf RePair} on input $a^m$ is 
\begin{center}
$S\to X_{\lfloor\log m\rfloor-1}X_{\lfloor\log m\rfloor-1}f_{\lfloor\log m\rfloor-1}(b_{\lfloor\log m\rfloor-1})\cdots f_1(b_1)f_0(b_0)$.
%\prod_{i=\lfloor\log m\rfloor-1}^0f_i(b_i)$
\end{center}
This means that the symbol $a$ only occurs in the start rule if $b_0=1$, and the nonterminal $X_i$ ($1\le i\le \lfloor\log m\rfloor-2$) occurs in the start rule if and only if
$b_i=1$. Since {\sf RePair} only replaces words with at least two occurrences, the most significant bit $b_{\lfloor \log m\rfloor}=1$ is represented by $X_{\lfloor\log m\rfloor-1}X_{\lfloor\log m\rfloor-1}$.
Note that for $1 \leq m \leq 3$, {\sf RePair} produces the trivial SLP $S \to a^m$.

\section{Main result}

The main result of this paper states:

\begin{theorem} \label{thm}
$\alpha_\mathsf{RePair}(2,n)\in\Omega\left(\log n/\log\log n\right)$ 
\end{theorem}
\begin{proof}
%We define a string $w_k$ of length $2k$ as follows.
We start with a binary De-Bruijn sequence $B_{\lceil\log k\rceil}\in \{0,1\}^*$ of length $2^{\lceil\log k\rceil}$
such that each factor of length $\lceil\log k\rceil$ occurs at most once % in $B_{\lceil\log k\rceil}$
\cite{deBr46}.
We have $k\le|B_{\lceil\log k\rceil}|< 2k$. Note that De-Bruijn sequences are not unique, so without loss of generality %w.l.o.g.
let us fix a De-Bruijn sequence which starts with $1$ for the remaining proof.
We define a homomorphism $h:\{0,1\}^*\to\{0,1\}^*$ by $h(0)=01$ and $h(1)=10$. The words $w_k$ of length $2k$ are defined as
$$w_k=h(B_{\lceil\log k\rceil}[1:k]).$$
For example for $k=4$ we can take $B_2=1100$, which yields $w_4=10100101$.
We will analyze the approximation ratio of {\sf RePair} for the binary words
$$s_k=\prod_{i=1}^{k-1}\left(a^{w_k[1:k+i]}b\right)a^{w_k}=a^{w_k[1:k+1]}ba^{w_k[1:k+2]}b\dots a^{w_k[1:2k-1]}ba^{w_k},$$
where the prefixes $w_k[1:k+i]$ for $1\le i\le k$ are interpreted as binary numbers. For example we have $s_4=a^{20}ba^{41}ba^{82}ba^{165}$.

Since $B_{\lceil\log k\rceil}[1]=w_k[1]=1$, we have $2^{k+i-1}\le\left|a^{w_k[1:k+i]}\right|\le 2^{k+i}-1$ for $1\le i\le k$ and thus $|s_k|\in \Theta\left(4^k\right)$.

\medskip
\noindent 
{\em Claim 1.} A smallest SLP producing $s_k$ has size $\mc O(k)$.

\medskip
\noindent

There is an SLP $\bb A$ of size $\mc O(k)$ for the first $a$-block $a^{w_k[1:k+1]}$ of length $\Theta(2^k)$. % by Theorem~\ref{lemma:folklore}, point~\ref{nlogn}. 
Let $A$ be the start nonterminal of $\bb A$.
For the second $a$-block $a^{w_k[1:k+2]}$ we only need one additional rule: If $w_k[k+2]=0$, then we can produce $a^{w_k[1:k+2]}$ by the fresh nonterminal $B$ using the rule $B\to AA$. Otherwise, if $w_k[k+2]=1$, then we use $B\to AAa$. The iteration of that process yields for each $a$-block only one additional rule of size at most $3$. If we replace the $a$-blocks in $s_k$ by nonterminals as described, then the resulting word has size $2k+1$ and hence $g(s_k)\in \mc O(k)$.

\medskip
\noindent 
{\em Claim 2.} The SLP produced by {\sf RePair} on input $s_k$ has size $\Omega(k^2/\log k)$.

\medskip
\noindent

On unary inputs of length $m$, the start rule produced by {\sf RePair} is strongly related to the binary encoding of $m$ as described above. %Example~\ref{unary}.
On input $s_k$, the algorithm starts to produce a start rule which is similarly related to the binary words $w_k[1:k+i]$ for $1\le i\le k$.
Consider the SLP $\mathbb{G}$ which is produced by {\sf RePair} after $(k-1)$ rounds on input $s_k$. We claim that up to this point {\sf RePair} is not affected by the $b$'s in $s_k$ and therefore has introduced the rules $X_1\to aa$ and $X_i\to X_{i-1}X_{i-1}$ for $2\le i\le k-1$. If this is true, then the start rule after $k-1$ rounds begins with
\begin{center}
$S\to X_{k-1}X_{k-1}f_{k-1}(w_k[2])f_{k-2}(w_k[3])\cdots f_0(w_k[k+1])b\cdots$
\end{center}
where $f_0(1)=a$, $f_0(0)=\varepsilon$ and $f_i(1)=X_i$, $f_i(0)=\varepsilon$ for $i\ge 1$.
All other $a$-blocks are longer than the first one, hence each factor of the start rule which corresponds to an $a$-block begins with $X_{k-1}X_{k-1}$. Therefore, the number of occurrences of $X_{k-1}X_{k-1}$ in the SLP is at least $k$. %, which is the number of $a$-blocks.
Since the symbol $b$ occurs only $k-1$ times in $s_k$, it follows that our assumption is correct and {\sf RePair} is not affected by the $b$'s in the first $(k-1)$ rounds on input $s_k$.
Also, for each block $a^{w_k[1:k+i]}$, the $k-1$ least significant bits of $w_k[1:k+i]$ ($1\le i\le k$) are represented in the corresponding factor of the start rule of $\mathbb{G}$, 
i.e., the start rule contains non-overlapping factors $v_i$ with
\begin{equation}
v_i=f_{k-2}(w_k[i+2])f_{k-3}(w_k[i+3])\dots f_1(w_k[k+i-1])f_0(w_k[k+i])\label{blockencoding}
\end{equation}
for $1\le i\le k$. For example after $3$ rounds on input $s_4=a^{20}ba^{41}ba^{82}ba^{165}$, we have the start rule
$$S\to \underbrace{X_3X_3X_2}_{a^{20}}b\underbrace{X_3^5a}_{a^{41}}b\underbrace{X_3^{10}X_1}_{a^{82}}b\underbrace{X_3^{20}X_2a}_{a^{165}},$$ 
where $v_1=X_2$, $v_2=a$, $v_3=X_1$ and $v_4=X_2a$.
The length of the factor $v_i\in\{a,X_1,\dots,X_{k-2}\}^*$ from equation~\eqref{blockencoding} is exactly the number of $1$'s in the word $w_k[i+2:k+i]$. Since $w_k$ is constructed by the homomorphism $h$, it is easy to see that $|v_i|\ge (k-3)/2$.
Note that no letter occurs more than once in $v_i$, hence $g(v_i)=|v_i|$. Further, each substring of length $2\lceil\log k\rceil+2$ occurs at most once in $v_1,\dots,v_k$, because otherwise there would be a factor of length $\lceil\log k\rceil$ occurring more than once in $B_{\lceil\log k\rceil}$. It follows that there are at least 
$$k\cdot ( \lceil(k-3)/2\rceil-2\lceil\log k\rceil-1)\in\Theta(k^2)$$
different factors of length $2\lceil\log k\rceil+2\in\Theta(\log k)$ in the right-hand side of the start rule of $\mathbb G$. By Lemma~\ref{lemma:folklore} it follows that a smallest SLP for the right-hand side of the start rule has size $\Omega(k^2/\log k)$ and therefore $|\mathsf{RePair}(s_k)|\in\Omega(k^2/\log k)$.

\medskip
\noindent 
In conclusion: We showed that a smallest SLP for $s_k$ has size $\mc O(k)$, while 
{\sf RePair} produces an SLP of size $\Omega(k^2/\log k)$.
This implies $\alpha_{\mathsf{RePair}}(s_k) \in \Omega(k/\log k)$, which together with $n=|s_k|$ and $k\in\Theta(\log n)$ finishes the proof.
\end{proof}
Note that in the above prove, {\sf RePair} chooses in the first $k-1$ rounds a digram for the replaced maximal string.
Therefore, Theorem~\ref{thm} also holds for the {\sf RePair}-variant, where in every round a digram (which is not necessarily
a maximal string) is replaced.

%\bibliographystyle{abbrv}
%\bibliography{../../../../BIBTEX/bib}

\end{document}